\documentclass[conference]{IEEEtran}

\usepackage{mathtools}
\usepackage[makeroom]{cancel}
\usepackage{amssymb}
\usepackage{graphicx}
\usepackage[english]{babel}
\usepackage{caption}
\usepackage{thmtools}
\usepackage{color}
\usepackage[font=footnotesize, margin=1cm]{caption}
\usepackage[framemethod=tikz]{mdframed}
\usepackage{lipsum}
\usepackage{comment}
\usepackage{mathabx}
\usepackage{bm}

{
\captionsetup[figure]{name=Fig.}

\pagestyle{empty}
	\DeclareCaptionStyle{mystylefigures}%
	[margin=5mm,justification=centering]%
	{font=footnotesize,margin={5mm,3mm}, labelsep = period}
	\captionsetup[figure]{style=mystylefigures}
\DeclareCaptionStyle{mystyletables}%
[margin=5mm,justification=centering]%
{font=footnotesize,labelfont=sc, textfont=sc, margin={8mm,5mm}, labelsep = period}
\captionsetup[table]{style=mystyletables}
\ifCLASSINFOpdf

\else

\fi

\newcommand{\nn}{\nonumber}
\newcommand{\mc}{\mathcal}

\newcommand{\mbb}{\mathbb}

\newcommand{\beq}{\begin{equation}}
\newcommand{\eeq}{\end{equation}}

\IEEEoverridecommandlockouts

\newcommand{\BEAS}{\begin{eqnarray*}}
\newcommand{\EEAS}{\end{eqnarray*}}
\newcommand{\BEQ}{\begin{equation}}
\newcommand{\EEQ}{\end{equation}}
\newcommand{\BIT}{\begin{itemize}}
\newcommand{\EIT}{\end{itemize}}

\newtheorem{theorem}{Theorem}
\newtheorem{lemma}{Lemma}

\newtheorem{corollary}{Corollary}

\usepackage{psfrag}

\usepackage{nomencl}

\makenomenclature

\usepackage{makeidx}
\makeindex

\usepackage{tikz}
\usetikzlibrary{shapes,backgrounds,calc}

\makeatletter
\tikzset{circle split part fill/.style  args={#1,#2}{%
 alias=tmp@name, 
  postaction={%
    insert path={
     \pgfextra{%
     \pgfpointdiff{\pgfpointanchor{\pgf@node@name}{center}}%
                  {\pgfpointanchor{\pgf@node@name}{east}}%
     \pgfmathsetmacro\insiderad{\pgf@x}
      \fill[#1] (\pgf@node@name.base) ([xshift=-\pgflinewidth]\pgf@node@name.east) arc
                          (0:180:\insiderad-\pgflinewidth)--cycle;
      \fill[#2] (\pgf@node@name.base) ([xshift=\pgflinewidth]\pgf@node@name.west)  arc
                           (180:360:\insiderad-\pgflinewidth)--cycle;            
         }}}}}
 \makeatother

\usepackage{pgffor}

\usepackage{algorithmic}
\usepackage{graphicx, subfig}
\usepackage{arydshln}

\begin{document}

\title{Ex-post Stable and Fair Payoff Allocation for Renewable Energy Aggregation}

\author{\IEEEauthorblockN{Hossein Khazaei 
and 
Yue Zhao 
}
\IEEEauthorblockA{
Department of Electrical and Computer Engineering,
Stony Brook University, Stony Brook, NY 11794, USA\\
Emails: \{hossein.khazaei, yue.zhao.2\}@stonybrook.edu
}
} 

\maketitle

\begin{abstract}
Aggregating statistically diverse renewable power producers (RPPs) is an effective way to reduce the uncertainty of the RPPs. The key question in aggregation of RPPs is how to allocate payoffs among the RPPs. In this paper, 
a payoff allocation mechanism (PAM) with a simple closed-form expression is proposed: It achieves stability (in the core) and fairness both in the ``ex-post'' sense, i.e., for all possible realizations of renewable power generation. 
Furthermore, this PAM can in fact be derived from the competitive equilibrium in a market. The proposed PAM is evaluated in a simulation study with ten wind power producers in the PJM interconnection. 

\end{abstract}

\section{Introduction}

Renewable energies play a central role in achieving a sustainable energy future. However, renewable energies such as wind and solar power are inherently uncertain and variable, and integrating them into the power system raises significant reliability and efficiency challenges \cite{NERC2009Report, Pinson2013}. A variety of approaches have been proposed to compensate for the uncertainty of renewable energies, such as employing better generation dispatch methods \cite{varaiya2011smart}, fast-ramping generators, energy storage control \cite{Bitar2011, Harsha14, CRZJG16}, and demand response programs \cite{Hobbs2012}. 

Another solution that has received considerable attention is to aggregate statistically diverse renewable energy sources \cite{NERC2009Report, Baeyens2013, ZQRGP15}. 
In an aggregation, renewable power producers (RPPs) pool their generation together so as to reduce the aggregate uncertainty and risk. As a result, by forming an aggregation, the RPPs can earn a higher payoff in total, e.g., in electricity markets. The central question in aggregating RPPs is then how to distribute the total payoff of an aggregation among its member RPPs. 
In particular, in designing payoff allocation mechanisms (PAMs), it is desired to ensure the stability of the aggregation (i.e., no subset of members are willing to leave), and the fairness of the payoffs. 
In general, achieving such properties can be studied in either the so-called ``ex-ante'' or the ``ex-post'' sense: In the ex-ante sense, stability and fairness properties are achieved ``in expectation''; In the ex-post sense, however, these properties must be achieved \emph{for all possible realizations of renewable power generation}. 

%
%
%

In this paper, 
we focus on achieving \emph{ex-post} stability and fairness. 
In particular, we show that a simple payoff allocation mechanism with a closed-form expression is ``in the core'' of the game in the ex-post sense, and thus guarantees the stability of the aggregation. 
We also show that this PAM achieves fairness properties in the ex-post sense as well. Furthermore, we show that the proposed PAM can in fact be derived from the \emph{competitive equilibrium} in a market. 

The remainder of the paper is organized as follows. In Section \ref{Relate}, related works on PAM design for achieving desirable properties in the ex-post sense are summarized. In Section \ref{ProblemFormul}, the structure of the two-settlement electricity market is described, and the desired properties of a PAM are defined. 
In Section \ref{ProposedInCorePAM}, we introduce a PAM that is proved to be in the core of the coalitional game. We also prove that this PAM achieves budget balance, individual rationality, fairness, and no-exploitation properties. In Section \ref{sec:market}, the proposed PAM is derived by computing the competitive equilibrium of a market with transferrable payoff. 
A case study using the data of ten wind power producers in PJM is given in Section \ref{Simul}. Conclusions are drawn in Section \ref{Concl}.

\section{Related Work on Payoff Allocation to Achieve Ex-post Properties} \label{Relate}

A PAM that achieves ``ex-post'' \emph{individual rationality, budget balance and fairness} is proposed in \cite{Nayyar13}. It is also proved that the resultant contract game of the model employed therein has at least one pure Nash equilibrium. A proportional cost sharing mechanism is proposed in \cite{LinBitar}. By neglecting the excess supply of the aggregator, the existence of at least one pure Nash equilibrium in the resultant contract game is investigated, and the properties of the mechanism are analyzed. 

An extension of the PAM in \cite{Nayyar13} is proposed in \cite{Harirchi14}, where statistical information of the RPPs is used to find the optimal forward power contract of the aggregator (similar to prior studies in the ``ex-ante'' settings \cite{Bitar2012}). The PAM of \cite{Nayyar13} is used as the base case, and the difference between the total payoff of the aggregation when it uses the optimal forward contract vs. when it uses the sum of RPPs' power contracts is distributed among the RPPs. The proposed PAM achieves budget balance, fairness and ex-ante individual rationality. 
A proportional cost sharing mechanism that achieves budget balance, no-exploitation and fairness is proposed in \cite{Harirchi16}. Also, assuming that the real-time selling price $p^{r,s}$ (cf. \eqref{PayoffSep} in Section \ref{sec:twoset}) is non-positive, it achieves ex-post individual rationality as well. 
It is then proved that the contract game among the RPPs, under certain assumption, has at least one pure Nash equilibrium. 

\section{Problem Formulation}   \label{ProblemFormul}
\subsection{Renewable Power Producers in a Two Settlement Market} \label{sec:twoset}

We consider RPPs participating in a two-settlement power market consisting of a day-ahead (DA) and a real time (RT) market. 

As a baseline case, we consider an RPP $i$ who participates in the market \emph{separately} from the other RPPs. In the DA market, RPP $i$'s generation at the time of interest in the next day is modeled as a random variable, denoted by $X_i$. RPP $i$ then sells a forward power supply contract in the amount of $c_i$ in the DA-market. It gets a payoff of $p^f c_i$ where $p^f$ denotes the price in the DA market. At the delivery time in the next day, 
RPP $i$ obtains its realized generation $x_i$: a) If it faces a shortfall, i.e., $c_i - x_i>0$, it needs to purchase the remaining power from the RT market at a real-time buying price $p^{r,b}$, b) if it has excess power, i.e. $x_i - c_i>0$, it can sell it in the RT-market at a real-time selling price $p^{r,s}$. In case excess power needs to be penalized as opposed to rewarded, we model such cases by having $p^{r,s}<0$. 
In this paper, the only assumption on prices that we make is $p^{r,s} \le p^{r,b}$, which must hold for no arbitrage. 

As a result, the \emph{realized} payoff of an RPP $i$ who separately participates in the market is given by 
\begin{align}   \label{PayoffSep}
\mathcal{P}_i^{sep} = p^f c_i - p^{r,b} \left(c_i - x_i\right)_+ + p^{r,s} \left(x_i - c_i\right)_+, 
\end{align}
where 
$(\cdot)_+ = \max(0,\cdot)$.

We note that, the DA market price $p^f$ can depend on the DA contracts $\{c_i\}$, and the real time prices $p^{r,b}$ and $p^{r,s}$ can depend on the realized power generation $\{x_i\}$ (see, e.g., \cite{ZQRGP15}). Other factors of uncertainty may also contribute to the price determination in the DA and RT markets. These complicate the choices of $c_i$ by an RPP, although with simplifying assumptions (e.g. price taking and known distribution of $x_i$) an optimal choice of $c_i$ can be computed (see  \cite{Bitar2012} among others). 

Nonetheless, the above complicating factors on price and contract determination do not affect our analysis for renewable energy aggregation in the ``ex-post'' sense, as will be shown next. 

\subsection{Aggregation of RPPs and Payoff Allocation} \label{GameDef}
Aggregation of RPPs reduces the risk of their uncertain future generation by exploiting the statistical diversity among the RPPs. 
We consider an aggregator that aggregates the power generation from a set of $N$ RPPs, denoted by $\mathcal{N}$, participates in the DA-RT market on behalf of them, and allocates its payoff back among the RPPs. In particular, we employ the following model (cf. \cite{Nayyar13, LinBitar, Harirchi14}) for aggregating RPPs: 
\begin{enumerate}
\item[a.] Each RPP $i$ submits a DA commitment $c_i ~(\ge0)$ to the aggregator. 
\item[b.] In the DA market, the aggregator sells a forward power contract in the amount of $c_{\mc{N}} = \sum_{i \in \mathcal{N}}{c_i}$. At the delivery time, the aggregator collects all the realized generation from the RPPs, denoted by $x_{\mc{N}} = \sum_{i \in \mathcal{N}}{x_i}$, to meet the commitment $c_{\mc{N}}$. The deviation  is settled in the RT market in the same way as in Section \ref{sec:twoset}. The realized payoff of the aggregator is thus given by
\begin{align}\label{PayoffAgg}
\mathcal{P}_A &= p^f c_{\mc{N}} - p^{r,b} \left(c_{\mc{N}} - x_{\mc{N}}\right)_+ + p^{r,s} \left(x_{\mc{N}} - c_{\mc{N}}\right)_+.
\end{align}
\item[c.] The aggregator returns a payoff of $\mc{P}_i$ to each RPP $i$. 
\end{enumerate}

The key design question for the aggregator lies in the step c. above, i.e., how to determine the payoffs $\mc{P}_i$. In particular, in  designing the payoff allocation mechanism (PAM), several desirable properties that an aggregator would like to achieve are the following: 
\begin{enumerate}
	\item \emph{Ex-post budget balance}: $\sum_{i \in \mathcal{N}}{\mathcal{P}_i} = \mathcal{P}_A$. 
	\item \emph{Ex-post individual rationality}: $\mathcal{P}_i \geq \mathcal{P}_i^{sep},  \forall i \in \mathcal{N}$. 
	\item \emph{Fairness}: For any two $\mbox{RPP}_i$ and $\mbox{RPP}_j$, if $c_i - x_i=c_j - x_j$, then $\mathcal{P}_i - p^f c_i = \mathcal{P}_j - p^f c_j$. 
	\item \emph{No-exploitation}: If $c_i - x_i=0$, then $\mathcal{P}_i = p^f c_i$. 
	\item \emph{Ex-post in the core}: Being in the core in an ``ex-post'' sense means that the RPPs' payoffs satisfy the following condition: 
	 if any subset of the RPPs leave the aggregation, separately form their own aggregation, and then participate in the market, they will get a payoff no higher than the sum of their payoffs originally from the aggregator. As a result, being in the core implies that the grand coalition/aggregation is stable.
\end{enumerate}
We note that achieving these properties in the ex-post sense means that they must hold for \emph{all possible realizations} of renewable power generation. 

Next, we make precise the definition of property 5), i.e., being in the core in the ex-post sense. 

\noindent{\bf A Coalitional Game and Its Core:} 
Similar to the realized payoffs \eqref{PayoffSep} and \eqref{PayoffAgg}, we define a function $v(\cdot)$ for the value of a coalition of a subset of RPPs $\mathcal{T}\subseteq\mc{N}$ as follows: 
\begin{align}  \label{ValOfCoal}
v \left(\mathcal{T}\right) &= p^f c_\mathcal{T} - p^{r,b} \left(c_\mathcal{T} - x_\mathcal{T}\right)_+ + p^{r,s} \left(x_\mathcal{T} - c_\mathcal{T}\right)_+ 
\end{align}
where $c_{\mathcal{T}} = \sum_{i \in \mathcal{T}}{c_i}$ and $x_{\mathcal{T}} = \sum_{i \in \mathcal{T}}{x_i}$. 
With the above function $v(\cdot)$, and a PAM $\{\mc{P}_i, i \in\mc{N}\}$, we have a well-defined coalitional game \cite{GameBook}. Now, in a coalitional game, a PAM $\{\mc{P}_i\}$ is said to be in the core if and only if it satisfies the following set of inequalities: 
	\begin{align}   \label{DefCoreGame}
	\forall \mc{T} \subseteq \mathcal{N}, ~ v\left(\mc{T}\right) \leq \sum_{i \in \mc{T}} \mc{P}_i. 
	\end{align}
Note that the intuition of \eqref{DefCoreGame} is exactly as described above in property 5). If \eqref{DefCoreGame} is satisfied for all possible realizations of the random renewable generation, the PAM $\{\mc{P}_i\}$ is said to be in the core in the ex-post sense. 

In Section \ref{ProposedInCorePAM}, we will introduce a PAM that achieves all the five properties above. Before we proceed to describing the proposed PAM, we close this section by providing a reasoning of the modeling assumption b. described above, namely, the aggregator sells a forward power contract in the amount of $c_{\mc{N}} = \sum_{i \in \mathcal{N}}{c_i}$. 


\subsection{The Necessity of $c_{\mc{N}} = \sum_{i \in \mathcal{N}}{c_i}$}    \label{CommentOnPowerConAgg}
Here we show that $c_{\mc{N}} = \sum_{i \in \mathcal{N}}{c_i}$ is a necessary condition for achieving property 2), i.e., ex-post individual rationality. It is thus also necessary for achieving property 5). 
We note that this is regardless of the choices of DA commitments $\{c_i\}$ that the RPPs submit to the aggregator. 

\begin{lemma}    \label{OptimalQuantityMustEqualSum}
	$\forall~\{c_i\}$, to achieve ex-post individual rationality and budget balance, we must have that $c_{\mc{N}} = \sum_{i \in \mathcal{N}}{c_i}$. 
\end{lemma}
\begin{proof}
Note that ex-post individual rationality implies that $\mathcal{P}_i \geq \mathcal{P}_i^{sep},  \forall i \in \mathcal{N}$ must hold for all possible realizations of  renewable power generation. 
We now prove that, if $c_{\mc{N}} > \sum_{i \in \mathcal{N}}{c_i}$, there exists a realization scenario in which it is impossible to achieve individual rationality. Such impossibility for the case of $c_{\mc{N}} < \sum_{i \in \mathcal{N}}{c_i}$ can be proven similarly. 
	
Suppose $c_{\mc{N}} > \sum_{i \in \mathcal{N}}{c_i}$. Consider the scenario that a) every RPPs has a shortfall, i.e., $\forall i \in \mathcal{N}, x_i < c_i$, and b) $p^f < p^{r,b}$. Consequently, 
	\begin{align}
	\mathcal{P}_A &= p^f c_{\mc{N}} - p^{r,b} \left(c_{\mc{N}} - \sum_{i \in \mathcal{N}} x_i\right),   \nonumber  \\
	\sum_{i \in \mathcal{N}} \mathcal{P}_i^{sep} &= \sum_{i \in \mathcal{N}} \left( p^f c_i - p^{r,b} \left(c_i - x_i\right) \right)  \nonumber    \\
~&= p^f \sum_{i \in \mathcal{N}} c_i - p^{r,b} \left(\sum_{i \in \mathcal{N}} c_i - \sum_{i \in \mathcal{N}} x_i\right),   \nonumber    \\
	\Rightarrow    \mathcal{P}_A & - \sum_{i \in \mathcal{N}} \mathcal{P}_i^{sep} = \left(p^f - p^{r,b}\right) \left(c_{\mc{N}} - \sum_{i \in \mathcal{N}} c_i\right) < 0,   \nonumber
	\end{align}
given that $p^f < p^{r,b}$. 
	With budget balance, $\mc{P}_A = \sum_{i\in\mc{N}}\mc{P}_i$. Thus, $\sum_{i\in\mc{N}}(\mc{P}_i -  \mathcal{P}_i^{sep}) < 0$, and it is impossible to achieve individual rationality (cf. property 2) in Section \ref{GameDef}) in this scenario.
\end{proof}

Now, with $c_{\mc{N}} = \sum_{i \in \mathcal{N}}{c_i}$, a desirable implication is that there is \emph{always} a non-negative excess profit from aggregating  the RPPs compared to having them participate in the market separately. 
\begin{corollary}\label{coro1}
$\forall~\{c_i\}$, with $c_{\mc{N}} = \sum_{i \in \mathcal{N}}{c_i}$, we have that $\mathcal{P}_A \ge \sum_{i \in \mathcal{N}} \mathcal{P}_i^{sep}$ for all possible realizations of renewable power generation. 
\end{corollary}

%

%
%
%
\section{Payoff Allocation Mechanism in the Core} \label{ProposedInCorePAM}
In this section, we introduce a payoff allocation mechanism that achieves all the five desired properties defined in Section \ref{GameDef}.

We first define the following notations for the (realization dependent) sets of RPPs with generation surpluses and shortfalls, respectively. 
\begin{align} \label{RPPsWithSurplussAndShortfall}
	&\mathcal{S}^+ := \left\{ i \in \mathcal{N} ~|~ x_i - c_i \geq 0 \right\}, ~ \mathcal{S}^- := \left\{ i \in \mathcal{N} ~|~ x_i - c_i < 0 \right\}    \nonumber \\
	&c_{\mathcal{S}^+} \! = \! \sum_{i \in \mathcal{S}^+} c_i , ~ x_{\mathcal{S}^+} \!=\! \sum_{i \in \mathcal{S}^+} x_i ,~ c_{\mathcal{S}^-} \!=\! \sum_{i \in \mathcal{S}^-} c_i , ~  x_{\mathcal{S}^-} \!=\! \sum_{i \in \mathcal{S}^-} x_i. 
\end{align}
We then have the following lemma on expressing the excess profit in terms of the above notations: 
\begin{lemma}
	The excess profit from aggregating the RPPs 
	\begin{align}   \label{ExcessProfit}
		\mathcal{P}_A &- \sum_{i \in \mathcal{N}} \mathcal{P}_i^{sep}    \nonumber \\
		&=\left(p^{r,b} - p^{r,s}\right) \min \left( \left(x_{\mathcal{S}^+} - c_{\mathcal{S}^+}\right) , \left(c_{\mathcal{S}^-} - x_{\mathcal{S}^-}\right) \right)
	\end{align}
\end{lemma}
\begin{proof}
First, we have that
	\begin{align}   \label{ExcssProfitSum}
	&\mathcal{P}_i^{sep} = \begin{cases}
	p^f c_i + p^{r,s} \left(x_i - c_i\right)       & \quad \text{if }  i \in \mathcal{S}^+ \\
	p^f c_i - p^{r,b} \left(c_i - x_i\right)       & \quad \text{if }  i \in \mathcal{S}^- \\
	\end{cases}   \nonumber  
	\end{align}
As a result, 
	\begin{align}
	&\sum_{i \in \mathcal{N}} \mathcal{P}_i^{sep} = \sum_{i \in \mathcal{S}^+} \mathcal{P}_i^{sep} + \sum_{i \in \mathcal{S}^-} \mathcal{P}_i^{sep}  \nonumber \\
	&=p^f \left(c_{\mathcal{S}^+} + c_{\mathcal{S}^-}\right) - p^{r,b} \left( c_{\mathcal{S}^-} - x_{\mathcal{S}^-} \right) + p^{r,s} \left( x_{\mathcal{S}^+} - c_{\mathcal{S}^+} \right) 
	\end{align}
%
We now consider the case of
		$x_{\mathcal{S}^+} - c_{\mathcal{S}^+}  \geq  c_{\mathcal{S}^-} - x_{\mathcal{S}^-}$, i.e., there is an excess power in total in the aggregation. In this case, 
		\begin{align}   \label{ExcssProfitAgg1}
		&\mathcal{P}_A = p^f \left(c_{\mathcal{S}^+} + c_{\mathcal{S}^-}\right) + p^{r,s} \left(x_{\mathcal{S}^+} + x_{\mathcal{S}^-} - c_{\mathcal{S}^+} - c_{\mathcal{S}^-}\right)
		\end{align}
		From (\ref{ExcssProfitSum}) and (\ref{ExcssProfitAgg1}), we have:
		\begin{align}
		&\mathcal{P}_A - \sum_{i \in \mathcal{N}} \mathcal{P}_i^{sep} = (p^{r,b} - p^{r,s}) \left(c_{\mathcal{S}^-} - x_{\mathcal{S}^-}\right) \nonumber \\
		&= \left(p^{r,b} - p^{r,s}\right) \min \left( \left(x_{\mathcal{S}^+} - c_{\mathcal{S}^+}\right) , \left(c_{\mathcal{S}^-} - x_{\mathcal{S}^-}\right) \right)
		\end{align}
		
		The case when $x_{\mathcal{S}^+} - c_{\mathcal{S}^+}  <  c_{\mathcal{S}^-} - x_{\mathcal{S}^-}$ can be proved similarly. 

\end{proof}

\subsection{The Proposed Payoff Allocation Mechanism}
We now introduce the main result of this paper, namely, the following payoff allocation mechanism: 
\begin{align}   \label{NewPAM}
	{\mc{P}}_i = 
	\begin{cases}
	p^f c_i +  		p^{r,b} \left(x_i - c_i\right)       & \quad \text{if }  x_{\mathcal{N}} - c_{\mathcal{N}} < 0 \\
	p^f c_i + p^* \left(x_i - c_i\right)       & \quad \text{if }  x_{\mathcal{N}} - c_{\mathcal{N}} = 0\\
p^f c_i +  		p^{r,s} \left(x_i - c_i\right)       & \quad \text{if }  x_{\mathcal{N}} - c_{\mathcal{N}} > 0 
	\end{cases},
\end{align}
where $p^{r,s}\le p^*\le p^{r,b}$, and $p^*$ can be chosen arbitrarily within this range. 

To understand the intuition of this PAM, it is instructive to consider the following two cases, respectively: 
\begin{itemize}
\item \emph{Case 1: The grand coalition has a shortfall in total, i.e. $x_{\mathcal{N}} - c_{\mathcal{N}} < 0$}. In this case, according to 
\eqref{NewPAM}, $\mc{P}_i = \mc{P}_i^{sep}, \forall i\in \mc{S}^-$, i.e., those RPPs with a shortfall earns exactly the same as if they each participates in the market separately. In comparison, $\forall i\in \mc{S}^+,~ \mc{P}_i - \mc{P}_i^{sep} = (p^{r,b} - p^{r,s})(x_i - c_i)$. As a result, only those RPPs in $S^+$ can gain extra earnings compared to if they participate in the markets separately. In other words, \emph{all the excess profit} \eqref{ExcessProfit} are allocated to those RPPs with a surplus. 

\item \emph{Case 2: The grand coalition has a surplus in total, i.e. $x_{\mathcal{N}} - c_{\mathcal{N}} > 0$}. In this case, $\mc{P}_i = \mc{P}_i^{sep}, \forall i\in \mc{S}^+$, i.e., those RPPs with a surplus earn exactly the same as if they participate in the market separately. In comparison, $\forall i\in \mc{S}^-, ~\mc{P}_i - \mc{P}_i^{sep} = (p^{r,b} - p^{r,s})(c_i - x_i)$. As a result, \emph{all the excess profit} \eqref{ExcessProfit} are allocated to those RPPs with a shortfall. 

\item \emph{Case 3: The grand coalition exactly meets its total commitment, i.e. $x_{\mathcal{N}} = c_{\mathcal{N}}$}. In this case, there is a family of payoff allocations that are all in the core: any $p^*$ such that $p^{r,s}\le p^*\le p^{r,b}$ works. Consequently, the PAM that is in the core is not unique in this sense. 
\end{itemize}


\subsection{The Proposed PAM is in the Core}

We now state the main theorem of the paper: 

\begin{theorem} \label{incorethm}
The proposed PAM (\ref{NewPAM}) satisfies the five properties 1) - 5) (cf. Section \ref{GameDef}) in the ex-post sense.
\end{theorem}
\begin{proof}
Budget balance, fairness and no exploitation can be verified straightforwardly. Next, we prove that the PAM is in the core in the ex-post sense \eqref{DefCoreGame}, which then implies individual rationality. 
	
	Consider $\mc{T}$ as an arbitrary subset of $\mathcal{N}$, and define ${{\mc{T}}^+} := \left\{ i \in {\mc{T}} ~|~ x_i - c_i \geq 0 \right\}$ and ${{\mc{T}}^-} := \left\{ i \in {\mc{T}} ~|~ x_i - c_i < 0 \right\}$. 
	We now consider the three scenarios of the grand coalition having a shortfall in total, or a surplus, or neither, respectively. 
	\begin{itemize}
	\item \emph{Case 1: The grand coalition has a shortfall in total, i.e. $x_{\mathcal{N}} - c_{\mathcal{N}} < 0 $}. 
	
	We have that, 
	\begin{align}  \label{CoreNewPAMpart1}
		\sum_{i \in {\mc{T}}} {\mc{P}}_i &\stackrel{(\ref{NewPAM})}{=} \sum_{i \in {\mc{T}}} p^f c_i \nonumber \\
		&+ p^{r,b} \left( \left( \sum_{i \in {{{\mc{T}}^+}}} \left(x_i - c_i\right) \right) - \left( \sum_{i \in {{{\mc{T}}^-}}} \left(c_i - x_i\right) \right) \right).
	\end{align}
	In comparison, if the subset ${\mc{T}}$ leaves the grand aggregation and participates in the market as a ``smaller'' aggregation, its payoff would be the following: 
	\begin{align}  \label{CoreNewPAMpart2}
		&v \left({\mc{T}}\right) = \sum_{i \in {\mc{T}}} p^f c_i \nonumber \\
		&	 + \begin{cases}
			-p^{r,b} \left( \left( \sum_{i \in {{{\mc{T}}^-}}} \left(c_i - x_i\right) \right) - \left( \sum_{i \in {{{\mc{T}}^+}}} \left(x_i - c_i\right) \right) \right)   \\
			\hspace{140pt} \text{ if }  x_{{\mc{T}}} - c_{{\mc{T}}} < 0 \\
			+p^{r,s} \left( \left( \sum_{i \in {{{\mc{T}}^+}}} \left(x_i - c_i\right) \right) - \left( \sum_{i \in {{{\mc{T}}^-}}} \left(c_i - x_i\right) \right) \right)   \\
			\hspace{140pt} \text{ if }  x_{{\mc{T}}} - c_{{\mc{T}}} \ge 0.
		\end{cases}
	\end{align}
	From (\ref{CoreNewPAMpart1}) and (\ref{CoreNewPAMpart2}), with $p^{r,b} \ge p^{r,s}$, 
	we have that $\sum_{i \in {\mc{T}}} {\mc{P}}_i \geq v \left({\mc{T}}\right)$.
	
	\item \emph{Case 2: The grand coalition has a surplus in total, i.e. $x_{\mathcal{N}} - c_{\mathcal{N}} > 0$}. 
	
	We have that, 
	\begin{align}  \label{CoreNewPAMpart3}
		\sum_{i \in {\mc{T}}} {\mc{P}}_i &\stackrel{(\ref{NewPAM})}{=} \sum_{i \in {\mc{T}}} p^f c_i \nonumber \\
		&+ p^{r,s} \left( \left( \sum_{i \in {{{\mc{T}}^+}}} \left(x_i - c_i\right) \right) - \left( \sum_{i \in {{{\mc{T}}^-}}} \left(c_i - x_i\right) \right) \right)
	\end{align}
	From (\ref{CoreNewPAMpart2}) and (\ref{CoreNewPAMpart3}), with $p^{r,b} \ge p^{r,s}$, we have that $\sum_{i \in {\mc{T}}} {\mc{P}}_i \geq v \left({\mc{T}}\right)$. 
	
	\item \emph{Case 3: The grand coalition exactly meets its total commitment, i.e. $x_{\mathcal{N}} = c_{\mathcal{N}}$}. 
	We have that, 
	\begin{align}  \label{CoreNewPAMpart4}
		\sum_{i \in {\mc{T}}} {\mc{P}}_i &\stackrel{(\ref{NewPAM})}{=} \sum_{i \in {\mc{T}}} p^f c_i \nonumber \\
		&+ p^{*} \left( \left( \sum_{i \in {{{\mc{T}}^+}}} \left(x_i - c_i\right) \right) - \left( \sum_{i \in {{{\mc{T}}^-}}} \left(c_i - x_i\right) \right) \right) 
	\end{align}
	From (\ref{CoreNewPAMpart2}) and (\ref{CoreNewPAMpart4}), with $p^{r,s} \le p^* \le p^{r,b} $, we have that $\sum_{i \in {\mc{T}}} {\mc{P}}_i \geq v \left({\mc{T}}\right)$.

	\end{itemize}
	
\end{proof}

\section{Deriving the Payoff Allocation Mechanism from a Competitive Equilibrium} \label{sec:market}
In this section, we show that the proposed PAM \eqref{NewPAM} can in fact be derived from computing the competitive equilibrium of a specially defined market. This is indeed how we discovered this PAM. The idea based on competitive equilibrium shares similar insight with that in a prior work \cite{ZQRGP15} which derives a closed-form PAM in the core in the ``ex-ante'' sense. 

\subsection{Market with Transferrable Payoff} \label{MTP}
We first define the following market with transferrable payoff \cite{GameBook}: 
\begin{itemize}
	\item The RPPs, denoted by $\mathcal{N}$, are a finite set of $N$ agents. 
	\item There is one type of input goods --- power generation. 
	\item Each agent $i \in \mathcal{N}$ has an ``endowment'' in the amount of $x_i \in \mathbb{R}_+$ --- the realized power of RPP $i$. 
	\item Each agent $i \in \mathcal{N}$ has a continuous, nondecreasing, and concave ``production'' function $f_i: \mathbb{R}_+ \rightarrow \mathbb{R}$:
	\begin{align}
	f_i(x_i) = \mathcal{P}_i^{sep} = p^f c_i - p^{r,b} \left(c_i - x_i\right)_+ + p^{r,s} \left(x_i - c_i\right)_+. \label{prodfun}
	\end{align}
\end{itemize}
Since all the ``production'' functions $\{f_i\}$ produce the same type of transferrable output, i.e., monetary payoff, the above formulation precisely defines a market with transferrable payoff. 

Next, a coalitional game can be defined based on a market with transferrable payoff \cite{GameBook}. Specifically, for any coalition of a subset of RPPs $\mathcal{T}\subseteq\mc{N}$, define 
\begin{align} \label{CoalValueFunMarketGame}
v \left(\mc{T}\right) = &\max_{\{z_i\in \mathbb{R}_+, i \in \mc{T}\}} ~ \sum_{i \in \mc{T}}f_i(z_i) \\
 &s.t.~ \sum_{i \in \mc{T}}{z_i} = \sum_{i \in \mc{T}}{x_i}. \nn
\end{align}
In other words, $\{z_i, i\in\mc{T}\}$ denotes a \emph{redistribution} of the total realized power $\sum_{i\in\mc{T}} x_i$ among the members of $\mc{T}$. This $v(\mc{T})$ represents the \emph{maximum} total payoff that the members of $\mc{T}$ can achieve among all possible redistributions, computed according to $f_i$ defined in \eqref{prodfun}. 
The core of this coalitional game is also called the ``core of the market''. 

We now prove that this coalitional game is exactly the same as the coalitional game defined previously in \eqref{ValOfCoal}. 
\begin{lemma} \label{equgame}
The values of coalitions \eqref{CoalValueFunMarketGame} are the same as \eqref{ValOfCoal}.
\end{lemma}
\begin{proof}
Straightforwardly, $\eqref{ValOfCoal} \ge  \eqref{CoalValueFunMarketGame}$ because \eqref{ValOfCoal} is the maximum achievable payoff by the subset $\mc{T}$ after their aggregation. Next, we show that \eqref{ValOfCoal} can be achieved by \eqref{CoalValueFunMarketGame}, i.e., $\eqref{ValOfCoal} \le  \eqref{CoalValueFunMarketGame}$. 

Similarly as in the proof of Theorem \ref{incorethm}, we define ${\mc{T}}^+ := \left\{ i \in {\mc{T}} ~|~ x_i - c_i \geq 0 \right\}$ and $\mc{T}^- := \left\{ i \in {\mc{T}} ~|~ x_i - c_i < 0 \right\}$. The intuition of a redistribution $\{z_i\}$ to achieve \eqref{ValOfCoal} is the following: We give as much of the \emph{excess} power of the RPPs in $\mc{T}^+$ as possible to the RPPs in $\mc{T}^-$ to offset their \emph{deficit} power. 

Specifically, 
if $x_{\mc{T}} - c_{\mc{T}} < 0$, i.e., $\sum_{i\in\mc{T}^-} \left(c_i - x_i\right) > \sum_{i\in\mc{T}^+} \left(x_i - c_i\right)$, we let 
\begin{align}
\forall i\in\mc{T}^+,& ~z_i = c_i, \label{zplus}\\
\forall i\in\mc{T}^-,& ~ x_i\le z_i \le c_i, \nn\\
&\mbox{ so that } \sum_{i\in\mc{T}^-} \left(z_i - x_i\right) = \sum_{i\in\mc{T}^+} \left(x_i - z_i\right). \label{zminus}
\end{align}
As a result, 
\begin{align}
\sum_{i \in \mc{T}}f_i(z_i) &= \sum_{i \in \mc{T}^+}f_i(z_i) + \sum_{i \in \mc{T}^-}f_i(z_i) \nn\\
&= \sum_{i \in \mc{T}^+}p^f c_i + \sum_{i \in \mc{T}^-}\left(p^fc_i - p^{r,b}(c_i-z_i) \right) \nn\\
& =p^f c_{\mc{T}} - p^{r,b}\sum_{i \in \mc{T}^-}\left((c_i-x_i) - (z_i - x_i) \right) \nn\\
& =p^f c_{\mc{T}} - p^{r,b}\left(\sum_{i \in \mc{T}^-}(c_i-x_i) - \sum_{i \in \mc{T}^+}(x_i - c_i) \right) \label{2ndtolast}\\
& = p^f c_{\mc{T}} - p^{r,b}(c_{\mc{T}} - x_{\mc{T}}) =  \eqref{ValOfCoal}, 
\end{align}
where \eqref{2ndtolast} is implied by \eqref{zplus} and \eqref{zminus}. 

The case of $x_{\mc{T}} - c_{\mc{T}} \ge 0$ can be proved similarly. 
\end{proof}

As a result, from the property of market with transferrable payoff (cf. Proposition 264.2 in \cite{GameBook}), we immediately have that this coaltional game has a \emph{non-empty core}. 

Moreover, this formulation as a market enables us to compute a solution in the core by deriving the \emph{competitive equilibrium} (CE) of the market, as follows. 

\subsection{Competitive Equilibrium} 
For the market with transferrable payoff defined in the last subsection, a competitive equilibrium is defined \cite{GameBook} as a price-quantity pair of $p^*\in\mbb{R}_+$ and $\bm{z}^*\in\mbb{R}_+^N$, such that, 
\begin{itemize}
\item[i)]
For each agent $i$, $z_i^*$ solves the following problem: 
\begin{align}  \label{OptZi}
\max_{z_i \in \mathbb{R}_+}{\left( f_i \left(z_i\right) - p^* \left(z_i - x_i\right) \right)}. 
\end{align}
\item[ii)] $\bm{z}^*$ is a redistribution, i.e., $\sum_{i\in\mc{N}}z_i^* = \sum_{i\in\mc{N}}x_i$. 
\end{itemize}
The intuition of a CE is the following: At the price $p^*$, i) to maximize its payoff, each agent $i$ can trade \emph{any} amount of the input (realized power) on the market \emph{without} worrying whether there is enough supply or demand to fulfill its trade request, and ii) collectively, the market  of input supply and demand \emph{still clears}, i.e., the resulting $\bm{z}^*$ from the optimal trades is feasible. 

At a competitive equilibrium $(p^*, \bm{z}^*)$, $p^*$ is called the \emph{competitive price}, and the value of the maximum of \eqref{OptZi} is called the \emph{competitive payoff} of agent $i$. 

We then have the following theorem (cf. Proposition 267.1 in \cite{GameBook}) dictating that all the CEs are in the core.
\begin{theorem} \label{thmCEcore}
Every profile of competitive payoffs in a market with transferable payoff is in the core of the market. 
\end{theorem}

Accordingly, to find a solution in the core of the market, which is also the core of the coalitional game for aggregating RPPs (cf. Lemma \ref{equgame}), it is sufficient to find a CE in the market with transferrable payoff defined in the last section. 
\vspace{5pt}

\noindent {\bf Deriving the Competitive Equilibrium:}

For the market with transferrable payoff defined in the last subsection, we have the following theorem: 
\begin{theorem} \label{thmCE}
Competitive equilibrium exists, and the competitive payoffs necessarily take the form of the proposed PAM \eqref{NewPAM}. 
\end{theorem}
\begin{proof}
With the production function $f_i(x_i)$ defined to be $\mathcal{P}_i^{sep}$ as in \eqref{prodfun}, we observe that $f_i(x_i)$ is a \emph{piecewise linear} function: $f_i'(x_i) = 
\begin{cases}
p^{r,b}, \mbox{ if } x_i < c_i \\
p^{r,s}, \mbox{ if } x_i > c_i
\end{cases}\!\!\!\!. $

As a result, at a CE, we must have $p^{r,b} \le p^* \le p^{r,s}$. Otherwise, by solving \eqref{OptZi}, either all RPPs would sell all of their power, or all of them would buy an infinite amount of power; Neither case would clear the market with $\sum_{i\in\mc{N}}z_i^* = \sum_{i\in\mc{N}}x_i$. 

We now analyze the optimal behavior of any agent $i$ under the following three scenarios of the competitive price $p^*$: 
\begin{itemize}
\item If $p^* = p^{r,b}$, 
the maximum of \eqref{OptZi} is achieved if and only if $z_i \le c_i$. 
\item If $p^* = p^{r,s}$, 
the maximum of \eqref{OptZi} is achieved if and only if $z_i \ge c_i$. 
\item If $p^{r,s} < p^* < p^{r,b}$,  
the maximum of \eqref{OptZi} is achieved if and only if $z_i = c_i$. 
\end{itemize}
%
%
%

To derive the competitive price $p^*$ that clears the market with $\sum_{i\in\mc{N}}z_i^* = \sum_{i\in\mc{N}}x_i$, we consider the following three scenarios: 

\emph{Case i)} $x_\mc{N} - c_\mc{N}<0$: 
As result, at the CE, $\sum_{i\in\mc{N}}z_i^* < \sum_{i\in\mc{N}}c_i^*$. From the above, we \emph{necessarily} have $p^* = p^{r,b}$. Indeed, with $p^* = p^{r,b}$, there exists $\bm{z}^*$ such that a) $z^*_i \le c_i$, and b) $\sum_{i\in\mc{N}}z_i^* = \sum_{i\in\mc{N}}x_i < c_\mc{N}$. 

Moreover, it is immediate to check that the competitive payoff of RPP $i$ equals $p^f c_i +  		p^{r,b} \left(x_i - c_i\right)$ (cf. \eqref{NewPAM}). 

\emph{Case ii)} $x_\mc{N} - c_\mc{N}>0$: 
As result, at the CE, $\sum_{i\in\mc{N}}z_i^* > \sum_{i\in\mc{N}}c_i^*$. From the above, we \emph{necessarily} have $p^* = p^{r,s}$. Indeed, with $p^* = p^{r,s}$, there exists $\bm{z}^*$ such that a) $z^*_i \ge c_i$, and b) $\sum_{i\in\mc{N}}z_i^* = \sum_{i\in\mc{N}}x_i > c_\mc{N}$. 

Moreover, the competitive payoff of RPP $i$ equals $p^f c_i +  		p^{r,s} \left(x_i - c_i\right)$ (cf. \eqref{NewPAM}). 

\emph{Case iii)} $x_\mc{N} - c_\mc{N}=0$: 
In this case, $\forall p^*,~s.t.~p^{r,s} \le p^* \le p^{r,b}$, $z_i^* = c_i, \forall i$ achieves $\sum_{i\in\mc{N}}z_i^* = \sum_{i\in\mc{N}}x_i = c_\mc{N}$. 

Moreover,  the competitive payoff of RPP $i$ equals $p^f c_i +  		p^* \left(x_i - c_i\right)$ (cf. \eqref{NewPAM}).
\end{proof}

From Theorem \ref{thmCEcore} and \ref{thmCE}, we conclude that the competitive payoffs that equal \eqref{NewPAM} are always in the core of the market, and hence the core of the coalitional game \eqref{ValOfCoal}. This thus offers an alternative proof for Theorem \ref{incorethm}.

\section{Simulation}  \label{Simul}
\subsection{Data Description and Simulation Setup}

We perform simulations using the NREL dataset \cite{NREL2010} based on ten wind power producers located in PJM. The simulation is run with the data of these ten wind power producers (WPPs) for the month of Feb. 2004. 
In each hour of the simulation, each WPP $i$ submits a day ahead forward power contract $c_i$: In particular, we assume that each WPP $i$ submits its \emph{optimal} contract as if it separately participates in the market. This optimal hourly day-ahead forward power contract can be solved as the solution to a \emph{news-vendor} problem \cite{Bitar2012}. To solve for this optimal individual contract, the statistical distribution of each WPP's future generation is needed. We consider that each WPP uses normal distributions to model its future generation: For each hour of the next day, a) an WPP's forecast for that future hour, available from the NREL dataset \cite{NREL2010}, is taken as the the mean of such a distribution, and b) the variance is approximated using the data of that WPP during Jan. 2004. 

\subsection{Simulation and Results}
For each hour, the aggregator a) commits in the DA market a forward power contract equal to the sum of the WPPs' forward power contracts, and b) collects the sum of the WPPs' realized generation. The payoff of the aggregator is computed according to \eqref{PayoffAgg} based on this sum of realization and its previously committed DA contract for this hour. The proposed PAM \eqref{NewPAM} is then employed to allocate the payoff to all the WPPs. 

For comparison, we evaluate the payoff that each WPP would get had it separately participate in the market \eqref{PayoffSep}. Furthermore, we compare the WPPs' payoffs with the payoff allocation mechanism developed in our prior work for the ``ex-ante'' setting \cite{Zhao2016}. 
The numerical results are summarized in Fig. \ref{fig1}. 
The sum of the payoffs of the ten WPPs for the three evaluations are compared in Table I. Finally, the traces of the hourly payoff over time for WPP 5 with the proposed PAM, the ex-ante PAM, and separate participation in the market are plotted in Fig. \ref{fig2}. 

\begin{figure}[h!]
	\centering
	\includegraphics[scale=0.5]{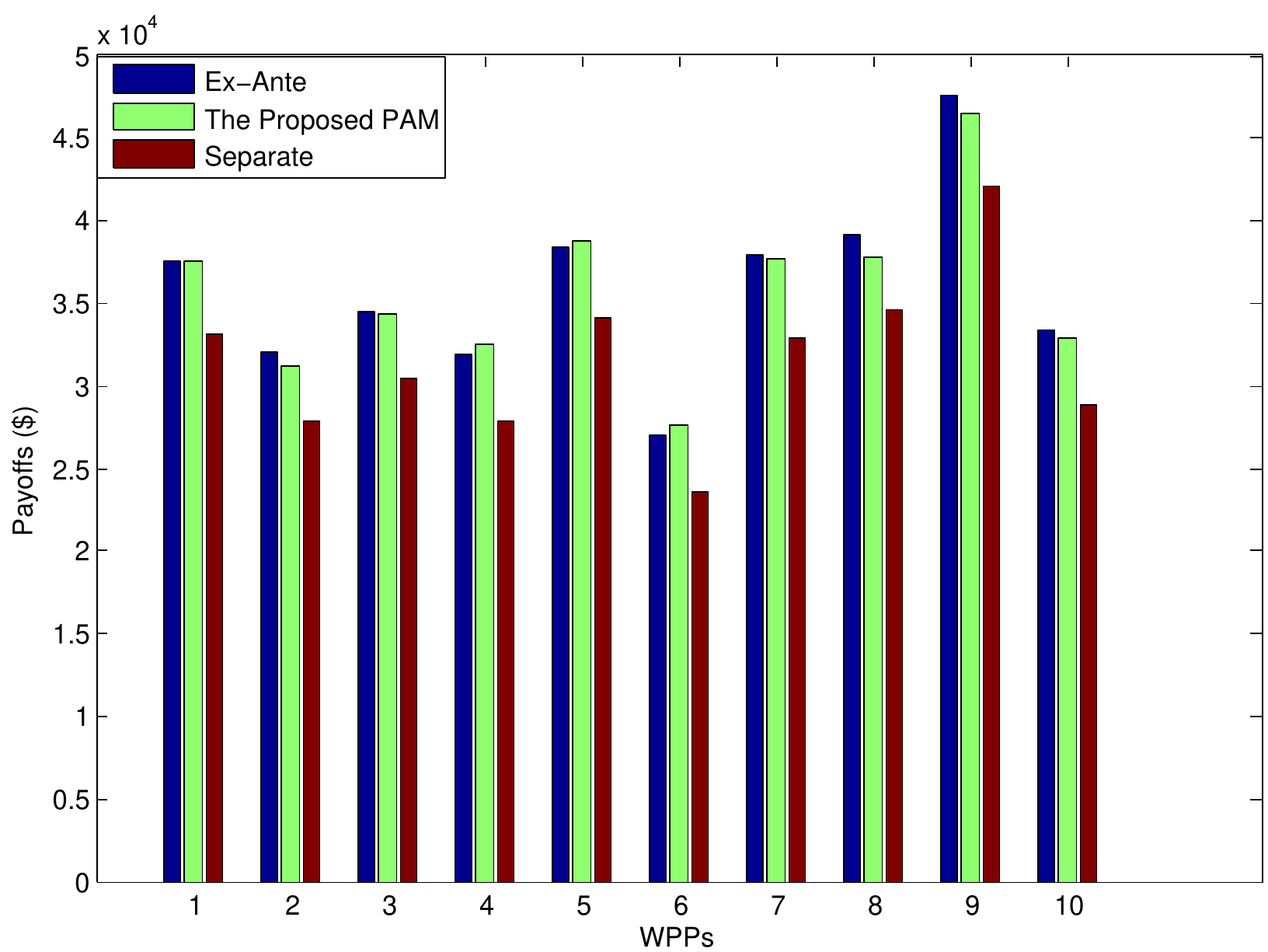}
	\caption{Comparison of daily average payoffs of the WPPs.}
	\label{fig1}
\end{figure}

We make the following observations:
\begin{itemize}
\item 
As expected, the proposed PAM achieves individual rationality in the ex-post sense. In contrast, the ``ex-ante'' PAM \cite{Zhao2016} achieves individual rationality only in expectation, and does not guarantee ex-post individual rationality. This can be observed around hour 145 in Fig. \ref{fig2}. 

Moreover, the proposed PAM, as proved, is also in the core in the ex-post sense. 

\item
As shown in Table \ref{tab1}, the total payoff of the WPPs with the proposed PAM is slightly $\left(0.72\%\right)$ lower than that with the ``ex-ante'' PAM \cite{Zhao2016}. The reason is the following. In the ``ex-ante'' PAM, the aggregator commits the \emph{optimal} power contract by solving the news-vendor problem \emph{for the grand coalition}. In contrast, in the model employed in this paper (cf. modeling assumptions a.-c. in Section \ref{GameDef}), the aggregator simply commits the sum of the WPPs' power contracts, which we assume are derived by solving the optimal forward contracts \emph{by each WPP individually} (and hence are not optimal from an aggregation's point of view). 

\begin{table}[h!]
	\centering
	\caption{Sum of the payoffs of the ten WPPs}
	\begin{tabular}{|c c c c|} 
		\hline
		& ``Ex-ante'' & The proposed & Separate\\ 
		& PAM \cite{Zhao2016} & \emph{in-the-core} PAM & ~ \\
		\hline
		Total      &  &  & \ \\ 
		payoff of &  $1.0428 \times 10^7$  &  $1.0353 \times 10^7$    &  $9.1480 \times 10^6$    \\
		the RPPs &    &     & \\
		\hline
	\end{tabular}
	\label{tab1}
\end{table}

\item
As shown in Fig. \ref{fig2}, the payoff traces of the proposed PAM and the ``ex-ante'' PAM follow each other very closely most of the times. However, the ``ex-ante'' PAM appears to be ``riskier''. The payoff with the ``ex-ante'' PAM sometimes has high peaks, higher than the corresponding payoff with the proposed PAM. An example is around hour 80. Meanwhile, the payoff with the ``ex-ante'' PAM sometimes also has low valleys, lower than the corresponding payoff with the proposed PAM. An example is around hour 580. This phenomenon of ``ex-ante'' PAM being riskier than the proposed is consistent with the fact that the ``ex-ante'' PAM is only individually rational in expectation, whereas the proposed PAM is in the core for all possible realizations of the renewable generation. 

\begin{figure}[h!]
	\centering
	\includegraphics[scale=0.4]{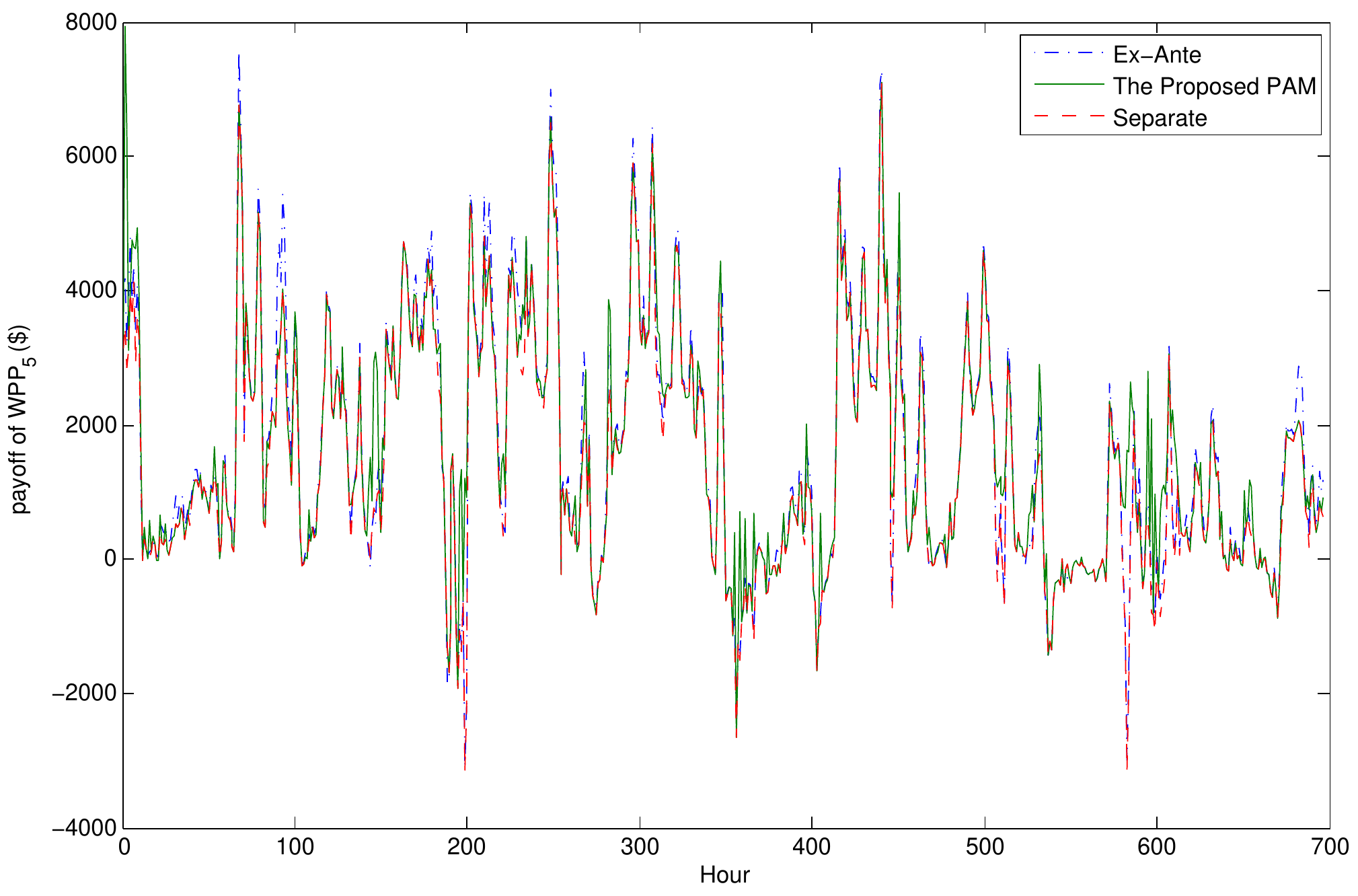}
	\caption{Comparison of payoffs of WPP $5$ for the proposed PAM, the ``ex-ante'' PAM \cite{Zhao2016}, and separate market participation.}
	\label{fig2}
\end{figure}
\end{itemize}

\section{Conclusion}    \label{Concl}
We have studied the problem of payoff allocation for aggregating renewable power producers (RPPs). We have first shown that the forward power contract of the aggregator must be equal to the sum of the member RPPs' power contracts in order to achieve ex-post individual rationality. We have then proposed a payoff allocation mechanism that is \emph{in the core} in the ex-post sense, and hence ensures the stability of the RPP aggregation. The proposed payoff allocation mechanism also achieves budget balance, individual rationality, fairness and no-exploitation, all in the ex-post sense. Moreover, we show that the proposed PAM can be derived from a competitive equilibrium in a market with transferrable payoff. 
%

\ifCLASSOPTIONcaptionsoff
  \newpage
\fi

\bibliographystyle{IEEEtran}
{\bibliography{TPS}}

\end{document}